\documentclass[a4paper,journal]{IEEEtran}
\IEEEoverridecommandlockouts

\newcommand\hmmax{0}
\newcommand\bmmax{0}

\usepackage{graphicx}
\usepackage{tikz}
\usepackage{pgfplots}
\usepackage{xcolor}
\usepackage{color, colortbl}
\usepackage{amsmath}
\usepackage{bbm}
\usepackage{amsfonts, amssymb}
\usepackage{mathtools}
\usepackage{enumerate}
\usepackage{cite}
\usepackage[nolist]{acronym}
\usepackage{amsmath}

\DeclareMathOperator*{\argmin}{argmin}

\usepackage{booktabs} 		
\usepackage{diagbox}		
\usepackage{upgreek}
\usepackage{bbm}
\usepackage{Files/bm}
\usepackage{Files/winsnotation}
\usepackage{Files/Symbols_OPL}
\usepackage{Files/LVcolours} 

\usepackage[ruled,vlined,noresetcount]{algorithm2e}
\usepackage{setspace}	 	
\usepackage{comment}
\usepackage{physics}
\usepackage{enumitem}

\newlist{myitemize}{itemize}{3}
\setlist[myitemize,1]{label=1.,leftmargin=1em}
\setlist[myitemize,2]{label=$\rightarrow$,leftmargin=0.75em}
\setlist[myitemize,3]{label=$\diamond$}

\usepackage{array}
\newcolumntype{C}[1]{>{\centering\arraybackslash}p{#1}}

\makeatletter
\def\endthebibliography{%
  \def\@noitemerr{\@latex@warning{Empty `thebibliography' environment}}%
  \endlist
}
\makeatother
\usepackage{amsthm}
\theoremstyle{definition}
\newtheorem{definition}{Definition}

\newtheorem{theorem}{Theorem}
\newtheorem{lemma}{Lemma}
\newtheorem{corollary}{Corollary}

\usepackage{algpseudocode,amsmath}

\newcommand{\PauliX}{\M{X}}
\newcommand{\PauliY}{\M{Y}}
\newcommand{\PauliZ}{\M{Z}}

\pgfplotsset{compat=1.17}

\begin{document}

\title{Spanning Tree Matching Decoder \\ for Quantum Surface Codes}



\author{Diego Forlivesi,~\IEEEmembership{Student~Member,~IEEE,}
Lorenzo~Valentini,~\IEEEmembership{Member,~IEEE,}
        and~Marco~Chiani,~\IEEEmembership{Fellow,~IEEE}
\thanks{The authors are with the Department of Electrical, Electronic, and Information Engineering ``Guglielmo Marconi'' and CNIT/WiLab, University of Bologna, V.le Risorgimento 2, 40136 Bologna, Italy. E-mail: \{diego.forlivesi2,lorenzo.valentini13, marco.chiani\}@unibo.it. 
Work funded in part by the European Union - Next Generation EU, PNRR project PRIN n. 2022JES5S2. 
}
}

\maketitle 
\markboth{}{Forlivesi, Valentini, Chiani: Spanning Tree Matching Decoder for Quantum Surface Codes}

\begin{acronym}
\small
\acro{AWGN}{additive white Gaussian noise}
\acro{BCH}{Bose–Chaudhuri–Hocquenghem}
\acro{CDF}{cumulative distribution function}
\acro{CRC}{cyclic redundancy code}
\acro{LDPC}{low-density parity-check}
\acro{ML}{maximum likelihood}
\acro{MWPM}{minimum weight perfect matching}
\acro{QECC}{quantum error correcting code}
\acro{PDF}{probability density function}
\acro{PMF}{probability mass function}
\acro{MPS}{matrix product state}
\acro{WEP}{weight enumerator polynomial}
\acro{WE}{weight enumerator}
\acro{BD}{bounded distance}
\acro{QLDPC}{quantum low density parity check}
\acro{CSS}{Calderbank, Shor, and Steane}
\acro{MST}{minimum spanning tree}
\acro{PruST}{pruned spanning tree}
\acro{RFire}{Rapid-Fire}
\acro{UF}{union-find}
\acro{LEMON}{library for efficient modeling and optimization in networks}
\acro{STM}{spanning tree matching}
\acro{i.i.d.}{independent identically distributed}

\end{acronym}
\setcounter{page}{1}

\begin{abstract}
 We introduce the \ac{STM} decoder for surface codes, which guarantees the error correction capability up to the code's designed distance by first employing an instance of the \acl{MST} on a subset of ancilla qubits within the lattice. Then, a perfect matching graph is simply obtained, by selecting the edges more likely to be faulty. A comparative analysis reveals that the \ac{STM} decoder, at the cost of a slight performance degradation, provides a substantial advantage in decoding time compared to the \ac{MWPM} decoder. Finally, we propose an even more simplified and faster algorithm, the \ac{RFire} decoder, designed for scenarios where decoding speed is a critical requirement. 

\end{abstract}

\begin{IEEEkeywords} Quantum Error Correcting Codes, Quantum Communications, Quantum Computing, Surface Codes
\end{IEEEkeywords}

\section{Introduction}

Leveraging the unique characteristics of quantum mechanics has considerably expanded the possibilities within the realm of information management, encompassing various domains such as sensing, processing, and communication \cite{Pre:18,CacCalVan:20, Zam:23}. The key hurdle in constructing a quantum computer is the inevitable presence of errors that, if not addressed, rapidly degrade quantum information. Therefore, error correction is crucial for meaningful quantum computation \cite{Sho:95, Got:09, NieChu:10, Pfi:23, ZorDePGio:23}. Surface codes are considered central to the architecture for the first generation of quantum computers, thanks to their high error thresholds, planar structure, and locality \cite{BraKit:98,  Iol:24}. The \acf{MWPM} decoder is presently the most widely used decoder for surface codes \cite{Higg:22, Hig:23, Bro:23}. This decoder results in large threshold error rates, but in practice, its polynomial time complexity introduces a latency that can be a bottleneck for fault-tolerant quantum computing architectures \cite{Got:09, kol:09, Ter:15}. Considerable efforts have been invested in optimizing the performance of this decoder \cite{Fow:13,Sko:23}. Sub-optimal algorithms based on the \ac{UF} decoder have also been proposed, achieving almost linear time in code length \cite{Del:20, Del:21}. Additionally, there is ongoing research into neural network-based solutions \cite{Gic:23}.


In this letter, we present a fast decoding technique tailored for surface codes, called the \acf{STM} decoder. The \ac{STM} algorithm involves implementing an instance of the \ac{MST} on a subset of the ancilla qubits within the lattice. This is followed by a simple and fast construction of a perfect matching graph, resulting in the estimated error pattern. Finally, we propose an even more simplified and faster algorithm, the RFire decoder, tailored for situations where decoding speed is of paramount importance. We conduct a comparative analysis of the performance, considering logical error rates and execution times, between the \ac{STM}, \ac{RFire}, and \ac{MWPM} decoders. We show that, at the price of some performance degradation, the proposed algorithms offer significant advantages in terms of decoding time.

\section{Preliminaries and Background}
\label{sec:preliminary}

\subsection{Quantum Stabilizer Error-Correcting Codes}

The Pauli operators are denoted as $\PauliX, \PauliY$, and $\PauliZ$. 
We indicate with $[[n,k,d]]$ a  \ac{QECC} encoding $k$ logical qubits $\ket{\varphi}$ into a codeword of $n$ data qubits $\ket{\psi}$, with minimum distance $d$. The code allows the correction of all patterns with up to $t = \lfloor(d-1)/2 \rfloor$ data qubit errors. 

Employing the stabilizer formalism, each code is characterized by $n-k$ independent and commuting operators $\M{G}_i \in \mathcal{G}_n$, termed stabilizer generators or simply generators, with $\mathcal{G}_n$ being the Pauli group on $n$ qubits \cite{Got:09, NieChu:10}.
The subgroup of $\mathcal{G}_n$ generated by all combinations of the $\M{G}_i \in \mathcal{G}_n$ is called stabilizer and indicated as $\mathcal{S}$. 
The code $\mathcal{C}$ is the set of quantum states $\ket{\psi}$ stabilized by $\mathcal{S}$, i.e., satisfying 
$\M{S}\ket{\psi}=\ket{\psi} \, \forall \M{S} \in \mathcal{S}$, or, equivalently, $\M{G}_i \ket{\psi}=\ket{\psi},\, i=1, 2, \ldots, n-k$.
The operators that commute with the stabilizer group but are not part of it are called logical operators.  
The generators specify measurements on quantum codewords that do not disturb the original quantum state. 
These measurements are carried out using extra qubits, named ancillas.
In fact, assume an error $\M{E}\in \mathcal{G}_n$ affects a codeword, so that the state becomes  $\M{E}\ket{\psi}$. It is possible to extract a binary sequence $\V{s}$ (also referred to as error syndrome) where the $i$-th entry $s_i$ is zero if $\M{G}_i$ commutes with $\M{E}$, while $s_i=1$ if $\M{G}_i$ anticommutes with it.
This enable the possibility to perform quantum error correction through error syndrome decoding using as input the binary sequence $\V{s}$.
We will refer to ancillas measuring $s_i = 1$ as \emph{defects}.

Among stabilizer codes we find the surface codes. These codes arrange qubits on a planar sheet \cite{BraKit:98, DenKitLan:02, HorFowDev:12, AtaTucBar:21}. 
One of the advantages of such an arrangement, is that it requires only nearest-neighbor interactions between qubits.
Also, this structure enables the possibility to perform a single round of stabilizer measurements with parallel operations~\cite{BluDolEve:23}. Logical operators can be easily identified on the surface codes: $\M{Z}_L$ ($\M{X}_L$) operator consists of a tensor product of $\M{Z}$'s ($\M{X}$'s) crossing horizontally (vertically) the lattice \cite{DenKitLan:02}. 
In the case of a \ac{QECC} over the depolarizing channel with data qubit error probability $p \ll 1$, we can approximate the logical error rate
as \cite{ForValChi24:MacW}
\begin{align}
\label{eq:error_probWithBetaApprox}
p_\mathrm{L} 
&\approx \left(1-\beta_{t+1}\right) \binom{n}{t+1}p^{t+1} 
\end{align}
where $\beta_{t+1}$ 
is the fraction of errors of weight $t+1$ that the decoder is able to correct.

\subsection{Minimum Weight Perfect Matching}
\label{subsec:MWPM}

A key characteristic of surface codes is the availability of a minimum weight decoder, known as the \ac{MWPM}. This syndrome decoder constructs a graph in which vertices represent defects, and edges are assigned weights depending on the error probability of the qubits separating the pair. For instance, in case of \ac{i.i.d.} data qubit errors, these weights correspond to the number of qubits between the pair.
In general, such a weighting procedure is performed using the Dijkstra algorithm to connect the defects on the graph representing the full lattice~\cite{Higg:22}.
By suitably assigning distances on the lattice, Dijkstra's algorithm allows for the consideration of different systems and qubit error statistics.
For surface and rotated surface structures~\cite{HorFowDev:12, ForValChi:23}, considering \ac{i.i.d.} data qubit errors, this step can be greatly simplified adopting the Manhattan distance between defects.
Then, the \ac{MWPM} proceeds by matching the graph to estimate the data qubit errors. A matching of a graph is a set of edges such that no two edges in the matching share a common vertex. A perfect matching is a matching that includes all vertices in the graph \cite{Die:08}.
For topological quantum codes without boundaries (toric codes), 
we can directly search for perfect matches using the \ac{MWPM} algorithm on the resulting graph \cite{kol:09}.
For surface codes, which have boundaries, it is necessary to add ghost defects (also known as ghost ancillas~\cite{Had:08}) before applying Dijkstra's algorithm.
This is because errors occurring along a boundary excite only one ancilla, leading to an odd number of defects which impedes running the \ac{MWPM}.
Therefore, for each defect, a corresponding ghost defect is considered by the decoder. 
In the final graph, these defects are also connected between themselves with zero distance in order to preserve the error correction capability of the code~\cite{Had:08}. 

The decoder based on the \ac{MWPM} is able to guarantee the error correction capability $t$.
In addition, it is also able to correct a significant amount of error patterns of weight greater than $t$. 
As an example, for the $[[9, 1, 3]]$ rotated surface code over the depolarizing channel, we have $\beta_2 = 0.5$, i.e., the decoder is able to correct $50\%$ of the error patterns of weight two \cite{ForValChi:23}.
As implemented in the \ac{LEMON}, the \ac{MWPM} algorithm complexity is $O(N M \, \log N)$, where $N$ and $M$ represent the number of vertices and edges in the graph, respectively \cite{DezBalJut:11}.
Specifically, for a toric code we have $N = n_\mathrm{d}$ and $M = \binom{n_\mathrm{d}}{2}$, while for a surface code $N = 2\,n_\mathrm{d}$ and $M = \binom{2\,n_\mathrm{d}}{2}$ due to ghost ancillas, where $n_\mathrm{d}$ denotes the number of defects.

Since planar architectures are favorable from an implementation point of view, we will focus on surface codes in the following. 
The \ac{STM} decoder for toric codes can be easily constructed by straightforward variations on the algorithm described below. 

\section{Spanning Tree Matching Decoder}
\label{sec:MSP}
In this Section, we introduce the \ac{STM} decoder for surface codes. 
First, from \eqref{eq:error_probWithBetaApprox} we observe that, for a code with given $n$ and $t$, the error correction capability is determined by the fraction of errors of weight $ w = t + 1$ that the decoder is able to correct.
Since the surface codes belong to the class of \ac{CSS} codes, for the sake of simplicity we will refer to the lattice where the sites constitute the $\M{X}$ generators. 
The same reasoning can be applied to the dual lattice. 

The \ac{STM} decoder consists of three phases: \ac{MST} evaluation, tree matching procedure, and error correction.

\subsection{Minimum Spanning Tree Phase}
\label{sec:MST}

We first construct the complete graph $\mathcal{G} = (n_\mathrm{d},\binom{n_\mathrm{d}}{2})$ connecting all defects, by using Dijkstra or the Manhattan distances. 
Note that, since there are no ghost ancillas at this stage, this graph has fewer edges and vertices than the one used by the \ac{MWPM}. 

Then, we execute the \ac{MST} algorithm on the graph.
This can be achieved with a complexity of $O(M\log N)$, where $M= \binom{n_\mathrm{d}}{2}$ and $N = n_\mathrm{d}$ stand for the edges and the vertices of the graph, respectively \cite{Cor:22}.
At this stage, we construct two \acp{MST} starting from the obtained one. 
If the number of defects is already even, one of the output trees is the original one, while the other is derived by introducing a ghost defect on both the left and right sides.
If the number of defects is odd, an output tree is obtained by adding a ghost defect to the left, while the alternative output tree is built by adding a ghost defect to the right. Then, the added ghost defects are connected to the nearest defect in the lattice. An example of this phase is reported in Fig.~\ref{Fig:MinSpanTree}(b).
In case of multiple defects with a ghost defect at the same distance, we select the one with the highest minimum distance to other non-ghost defects.
\begin{figure}[t]
 	\centering
 	\resizebox{\columnwidth}{!}{
 	   \input{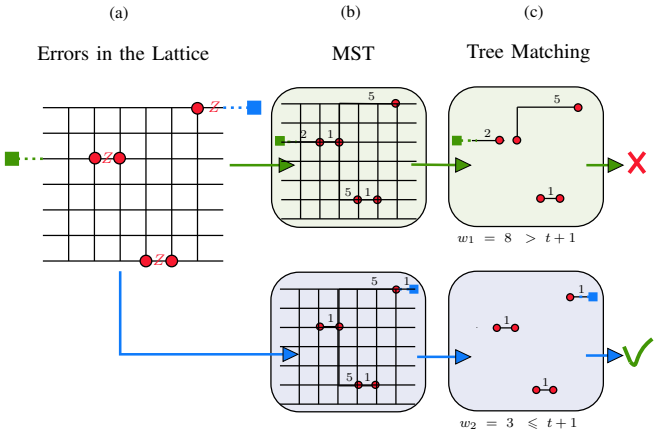}
     }
 	\caption{ Spanning tree matching decoder with a $[[85,1,7]]$ surface code. a) Three $\M{Z}$ channel errors occur on the lattice. Exited ancillas are depicted in red. b) Two alternative \acp{MST} obtained with the nearest ghost ancilla to the left (above) and to the right (below) boundary, respectively. c) Resulting $\mathcal{E}$ from the tree matching procedure.}
 	\label{Fig:MinSpanTree}
\end{figure}
%


\subsection{Tree Matching Phase}
\label{sec:Pru}

During this step, we match the trees to obtain the estimated error patterns. For a tree $\mathcal{T}$ let us define as $\mathcal{E}$ its perfect matching graph.
We use $\text{adj}(v, \mathcal{G})$ to indicate the set of vertices adjacent to the vertex $v$ in the graph $\mathcal{G}$, and $\text{deg}(v, \mathcal{G})$ for the degree of $v$ in the graph $\mathcal{G}$. $\mathcal{B}(\mathcal{T})=\{v \in \mathcal{T} \, | \,\text{deg}(v, \mathcal{T})=1\}$ is the set of boundary vertices in  $\mathcal{T}$.

A simple algorithm for obtaining the perfect matching of an \ac{MST} consists of the following steps. 
\begin{myitemize}
    \item[1.] For each $b \in \mathcal{B}(\mathcal{T})$ 
    \begin{myitemize}
        \item[] Let $a=\text{adj}(b, \mathcal{T})$;
        \item[] If $\text{deg}(a, \mathcal{T}) = 2$: call the edges $e_1 = (a, b)$ and $e_2 = (a, \, \text{adj}(a, \mathcal{T}) \setminus \{b\})$. Add $e_1$ to $\mathcal{E}$, and remove the subgraph $(\{ a, b\}, \{e_1, e_2\})$ from $\mathcal{T}$.
    \end{myitemize}
    
    \item[2.] For each $b \in \mathcal{B}(\mathcal{T})$ 
    \begin{myitemize}
        \item[] Let $a=\text{adj}(b, \mathcal{T})$;
    
        \item[] If $\text{deg}(a, \mathcal{T}) = 3$: call $\{v_1,v_2\}= \text{adj}(a, \mathcal{T}) \setminus \{b\}$, and the edges $e_1 = (a, b)$, $e_2 = (a, v_1)$, $e_3 = (a, v_2)$. Add the edge $e_1$ to $\mathcal{E}$ and remove from $\mathcal{T}$ the subgraph $(\{ a, b\}, \{e_1, e_2, e_3\})$. Insert the edge $e = (v_1, v_2)$ with weight $w(e) = w(e_2) + w(e_3)$ in $\mathcal{T}$. Return to step 1.
    \end{myitemize}
    \item[3.] For each $b \in \mathcal{B}(\mathcal{T})$ 
    \begin{myitemize}
    \item[] Let $a=\text{adj}(b, \mathcal{T})$;
    
    \item[] If $\text{deg}(a, \mathcal{T}) = 4$: remove from $\mathcal{T}$ the edge $e = (a, v)$ where $v$ is the sole vertex in $\text{adj}(a, \mathcal{T})$ with $\text{deg}(v, \mathcal{T}) > 1$. Return to step 2.
    \end{myitemize}
\end{myitemize}

Note that, considering \ac{i.i.d.} data qubit errors, a vertex $v$ in the \ac{MST} cannot have $\text{deg}(v, \mathcal{T}) > 4$.

The algorithm terminates when the graph $\mathcal{T}$ becomes empty, and gives a set of edges $\mathcal{E}$ representing a valid solution to the error correction problem. This procedure is carried out for both the \acp{MST} obtained in Section~\ref{sec:MST}, and the two solutions are indicated as $\mathcal{E}_1$ and $\mathcal{E}_2$.
Examples of the tree matching procedure described above are depicted in Fig.~\ref{Fig:MinSpanTree}(c) and in Fig.~\ref{Fig:MinSpanTreee}(c). 

\subsection{Error Correction}
\label{sec:CEC}

After the previous phase we have two possible sets of faulty qubits, $\mathcal{E}_1$, $\mathcal{E}_2$, with a total number of data qubit errors $w_1$ and $w_2$, respectively. An instance of \ac{STM} decoding, involving both \acp{MST}, is illustrated in Fig.~\ref{Fig:MinSpanTree}.
If $w_1 \le t+1$ or $w_2 \le t+1$, then we have found the correction operator\footnote{Note that, if $w_1 \leq t$, the matching of the second tree can be avoided to save time.}. 

If, instead, both $w_1$ and $w_2$ have weight $> t+1$, a simple processing considering both $\mathcal{E}_1$ and $\mathcal{E}_2$ can be applied, which guarantees the correction if the channel errors has weight $\le t$ (see  Section~\ref{subsec:DPD}).  
The basic idea is choosing the solution with the smallest number of horizontally traversed edges.
This is motivated by the fact that logical operators traverse the lattice from one side to the other one.
Hence, the correction with more operators along the horizontal dimension will more likely cause a logical operator (i.e., an undetected error).
\begin{definition}
    A column is the set of horizontal edges aligned in the vertical direction of the lattice, as shown in Fig.~\ref{Fig:MinSpanTreee}(a).
\end{definition}
In this way, we can enumerate the columns from left to right ranging from $1$ to $d$.
Then, we can define two vectors $\V{c}_{i}$ with entries $c_{i,j}$, where $i=1, 2$ and $j=1, \dots, d$, representing the cardinality of the intersection between $\mathcal{E}_i$ and the $j$-th column.
Since a solution $\mathcal{E}_i$ with two edges in the $j$-th column is equivalent, by adding a stabilizer, to another solution without any edge in the $j$-th column, 
we also define as $\V{u}_{i}$ a vector with entries $u_{i,j} = c_{i,j} \mod 2$.
As a consequence, the function
\begin{align}
    \label{eq:metricCol}
    f(\mathcal{E}_i) = \sum_{j=1}^{d} u_{i,j}
\end{align}
can be used as a metric to quantify how much correcting $\mathcal{E}_i$ is likely to cause a logical operator.
Thus, the final solution is that minimizing \eqref{eq:metricCol}.
For instance, in Fig.~\ref{Fig:MinSpanTreee}(c) there are three edges in solution $\mathcal{E}_1$, with weights $w = 1$, $w = 2$, and $w = 5$. They consist of one, two, and three horizontal qubits, respectively. However, the horizontal qubit of the edge with weight $w = 1$ belongs to the same column as the edge with weight $w = 5$. Hence, $f(\mathcal{E}_i) = 4$.

 \begin{figure}
 	\centering
 	\resizebox{\columnwidth}{!}{
 	   \input{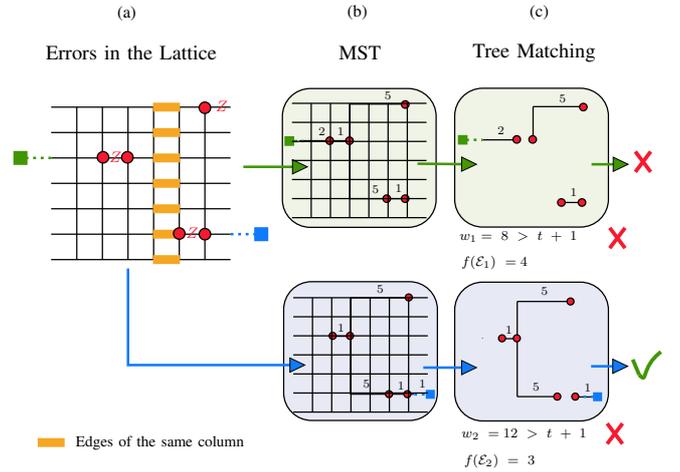}
     }
 	\caption{ Spanning tree matching decoder with a $[[85,1,7]]$ surface code. Both matched spanning trees have $w > t + 1$. Hence, the error correction is performed according to \eqref{eq:metricCol}.    }
 	\label{Fig:MinSpanTreee}
 \end{figure}

\subsection{Distance-preserving decoding}
\label{subsec:DPD}

In this section we show that \eqref{eq:metricCol}, in the case of a $[[n,k,d=2t +~1]]$ surface code, assures the correction of channel errors with weight $w \leq t$. 
\begin{lemma}
\label{thm:minpath}
Given a surface code, let us call $\mathcal{C}$ an arbitrary Pauli $\M{Z}$ error chain connecting two sites defects $v_1$ and $v_2$. Then, the intersection of $\mathcal{C}$ with any column between $v_1$ and $v_2$ has cardinality $1 \pmod 2$. The intersection with any of the other columns has cardinality $0 \pmod 2$. 
\end{lemma}
\begin{proof}
The minimal-weight chain linking $v_1$ and $v_2$ clearly satisfies the theorem. Every other chain can be obtained from this one, through the application of plaquette generators. Since each plaquette has two edges in the same column, it does not affect the modulo 2 counting. 
\end{proof}


\begin{lemma}
\label{lem:minmatch}
Given any error pattern over an $[[n,k,d]]$ surface code, resulting in an even number of site defects (including also ghost defects), all possible perfect matchings $\mathcal{M}$ have the same metric $f(\mathcal{M})$.
\end{lemma}
\begin{proof}
    We label the horizontal position of the sites in the lattice, from left to right, as $0, 1, \ldots, d+1$, with $0$ for the left ghost defect and $d+1$ for the right ghost defect. 
    Vertically aligned sites share the same label. 
    Let us firstly examine the case where the lattice has four site defects. 
    Reordering them we obtain four indexes $0 \le j_1 \le j_2 \le j_3 \le j_4 \le d+1$, representing the site positions. 
    In this setup we can have three possible perfect matchings $\mathcal{M}_i$, with $i=1,2,3$. 
    Considering  $\mathcal{M}_1$ as the one connecting site $1$ with $2$, and site $3$ with $4$, by application of Lemma~\ref{thm:minpath}, we have that $f(\mathcal{M}_1) = (j_2 - j_1) + (j_4 - j_3)$. It is easy to check that the other two matchings have the same metric. 
    Finally, we observe that each perfect matching $\mathcal{M}$ in any graph $\mathcal{G}$ with an even number of defects can be obtained by iteratevely removing two edges and reconnecting them as needed. 
    Hence, the claim follows since this operation does not affect the metric $f(\mathcal{M})$.
\end{proof}



\begin{corollary}
   \label{cor:min}
   Each perfect matching has the same metric \eqref{eq:metricCol} as the \ac{MWPM}.
    It follows that \eqref{eq:metricCol} can be interpreted as the minimum number of traversed columns for all perfect matchings of the same graph.
\end{corollary}

We have seen in Section~\ref{sec:MST} that, starting from a complete graph $\mathcal{G}$ formed by connecting all defects, we obtain two distinct graphs $\mathcal{G}_1$ and $\mathcal{G}_2$ by adding ghost ancillas. Due to the even number of defects in these graphs, we can compute two perfect matchings, denoted as $\mathcal{E}_1$ and $\mathcal{E}_2$. It is noteworthy that applying a logical operator to one matching we get the other. Consequently, one of the solutions will correct the error (call it $\mathcal{E}_c$), while the alternative, $\mathcal{E}_a$, will produce a logical error.

\begin{theorem}
\label{th:MinTravDistBelowT}
    Let us consider an $[[n,k,d]]$ surface code and an error pattern of weight $w \le t$. A perfect matching leading to the correct solution, $\mathcal{E}_c$, 
    has $f(\mathcal{E}_c) < f(\mathcal{E}_a)$, with $\mathcal{E}_a$ being any perfect matching on the alternative graph.
\end{theorem}

\begin{proof}
    The number of columns in the lattice is $d$.
    Since we are considering error patterns with weight $w \le t$, we have that any perfect matching $\mathcal{E}_\mathrm{c}$ representing the correct solution, due to Lemma~\ref{lem:minmatch}, satisfies $f(\mathcal{E}_\mathrm{c}) \le t$.
    Regarding the alternative solution $\mathcal{E}_\mathrm{a}$, we have that $f(\mathcal{E}_\mathrm{a}) = d - f(\mathcal{E}_\mathrm{c})$, since they differ by a logical operator. 
    This leads to $f(\mathcal{E}_\mathrm{a}) \ge t+1$ for $d$ odd and $f(\mathcal{E}_\mathrm{a}) \ge t+2$ for $d$ even,  proving the statement.
\end{proof}
    The previous theorem states that a surface decoder chosing a matching with the minimum traversed column metric \eqref{eq:metricCol} preserves the error correction capability of the code. 

    The same theorem applies to rotated surface codes as well, owing to their similar lattice structure.


\section{Rapid-fire decoder}\label{sec:RFDec}

We show now how to design an even faster decoder that ensures the correction for all errors with weights $w \leq t$. 
Due to Corollary~\ref{cor:min}, each matching is equivalent to the minimum one in terms of \eqref{eq:metricCol}.
Hence, it is possible to compute $\mathcal{E}_i$ with $i=1,2$ without evaluating and matching the \ac{MST}. 
Firstly, we compute the complete graph $\mathcal{G}$ on the defects. 
Then, we construct two graphs $\mathcal{S}_i$ by adding to $\mathcal{G}$ ghost ancillas on boundaries with the same strategy described in Sections~\ref{sec:MST}.
For each $\mathcal{S}_i$, we iteratively pair each defect with its closest one in a greedy fashion, we update the solution $\mathcal{E}_i$, and we remove both the defects from $\mathcal{S}_i$. 
This process results in two potential sets of faulty qubits, $\mathcal{E}_1$ and $\mathcal{E}_2$. 
Finally, we determine the error correction operator to adopt, following the procedure outlined in Section~\ref{sec:CEC}.
In this way, we guarantee the correction of all errors of weight up to $w = t$ by Theorem~\ref{th:MinTravDistBelowT}. 
We call this the \ac{RFire} decoder. 

\section{Numerical Results}\label{sec:NumRes}

\begin{figure}[t]
	\centering
	\resizebox{0.49\textwidth}{!}{ 
%
%
\definecolor{mycolor1}{rgb}{0.00000,0.44700,0.74100}%
\definecolor{mycolor2}{rgb}{0.85000,0.32500,0.09800}%
\definecolor{mycolor3}{rgb}{0.92900,0.69400,0.12500}%
\definecolor{mycolor4}{rgb}{0.49400,0.18400,0.55600}%
\definecolor{mycolor5}{rgb}{0.46600,0.67400,0.18800}%
\definecolor{mycolor6}{rgb}{0.30100,0.74500,0.93300}%
\begin{tikzpicture}

\begin{axis}[%
name = plot,
width=4.5in,
height=3.5in,
at={(0in,0in)},
scale only axis,
xmode=log,
xmin=0.001,
xmax=0.06,
xminorticks=true,
xlabel style={font=\color{white!15!black}, font = \Large},
xlabel={ $p$},
tick label style={black, semithick, font=\Large},
ymode=log,
ymin=1e-06,
ymax=0.1,
yminorticks=true,
ylabel style={font=\color{white!15!black}, font = \Large},
ylabel={ $p_L$},
axis background/.style={fill=white},
legend style={legend cell align=left, align=left, draw=white!15!black}
]
\addplot [color=brightRed, line width=1.5pt, mark=o, mark size = 2.5pt, mark repeat = 2, mark phase = 2 mark options={solid, brightRed}]
  table[row sep=crcr]{%
0.001	2e-05\\
0.005	0.00049\\
0.008	0.00118\\
0.01	0.00184\\
0.05	0.03946\\
}; \label{mst:3x3}

\addplot [color=brightRed, line width=1.5pt, dashed]
  table[row sep=crcr]{%
0.001	1.64e-05\\
0.005	0.00039113\\
0.008	0.000997308\\
0.01	0.001562046\\
0.05	0.0352298\\
}; \label{mst9:3x3}

\addplot [color=brightRed, line width=1.5pt, mark=square, mark size = 2.5pt, mark options={solid, brightRed}]
  table[row sep=crcr]{%
0.003	4.94286e-06\\
0.005	2.15675e-05\\
0.008	9.10919e-05\\
0.01	0.000172973\\
0.05	0.026\\
}; \label{mst:5x5}

\addplot [color=brightRed, line width=1.5pt, mark=diamond, mark size = 4pt, mark options={solid, brightRed}]
  table[row sep=crcr]{%
0.005   2.1e-6\\
0.008	2e-05\\
0.02	0.001\\
0.05	0.04\\
};  \label{mst:7x7}

\addplot [color=brightBlue, line width=1.5pt, mark=o, mark size = 2.5pt, mark repeat = 2, mark options={solid, brightBlue}]
  table[row sep=crcr]{%
0.001	2e-05\\
0.005	0.00046\\
0.008	0.00114\\
0.01	0.0017\\
0.05	0.03819\\
};  \label{mwpm:3x3}

\addplot [color=brightBlue, line width=1.5pt, dashed]
  table[row sep=crcr]{%
0.001	1.54e-05\\
0.005	0.0003926\\
0.008	0.0010174\\
0.01	0.0015618\\
0.05	0.0348174\\
}; \label{mwpm9:3x3}

\addplot [color= graphYellow, line width=1.5pt, dashed]
  table[row sep=crcr]{%
0.001	1.7e-05\\
0.005	0.000427084\\
0.008	0.001063158\\
0.01	0.001675\\
0.05	0.0357\\
}; \label{NOmst9:3x3}

\addplot [color=brightBlue, line width=1.5pt, mark=square, mark size = 2.5pt, mark options={solid, brightBlue}]
  table[row sep=crcr]{%
0.003	3.8e-06\\
0.005	1.8e-05\\
0.008	7.1e-05\\
0.01	0.000124\\
0.05	0.01565\\
};  \label{mwpm:5x5}

\addplot [color=brightBlue, line width=1.5pt, mark=diamond, mark size = 4pt, mark options={solid, brightBlue}]
  table[row sep=crcr]{%
0.008	3e-06\\
0.02	0.00015\\
0.05	0.0063\\
};  \label{mwpm:7x7}

\addplot [color=graphYellow, line width=1.5pt, mark=o, mark size = 2.5pt, mark options={solid, graphYellow}]
  table[row sep=crcr]{%
0.001	2.2e-05\\
0.005	0.0005025\\
0.008	0.00127778\\
0.01	0.00190909\\
0.05	0.0406\\
};  \label{NOmst:3x3}

\addplot [color=graphYellow, line width=1.5pt, mark=square, mark size = 2.5 pt,  mark options={solid, graphYellow}]
  table[row sep=crcr]{%
0.003	0.62e-05\\
0.005	2.7e-05\\
0.008	0.000108\\
0.01	0.000205\\
0.05	0.0312\\
}; \label{NOmst:5x5}

\addplot [color=graphYellow, line width=1.5pt, mark=diamond, mark size = 4 pt, mark options={solid, graphYellow}]
  table[row sep=crcr]{%
0.005   5.73298e-06 \\
0.008	3.41702e-05\\
0.02	0.00138\\
0.05	0.0462\\
}; \label{NOmst:7x7}

\end{axis}

\node[draw,fill=white,inner sep=1.5pt, below right=0.5em] at (plot.north west){
\large
  {\renewcommand{\arraystretch}{1.2}
    \begin{tabular}{lccc}
     & \textbf{RFire} &  \textbf{STM} & \textbf{MWPM} \\
     \textbf{$[[9,1,3]]$ } & \ref{NOmst9:3x3} & \ref{mst9:3x3} & \ref{mwpm9:3x3} \\
     \textbf{$[[13,1,3]]$ } & \ref{NOmst:3x3} & \ref{mst:3x3} & \ref{mwpm:3x3} \\
    \textbf{$[[41,1,5]]$} & \ref{NOmst:5x5} & \ref{mst:5x5} & \ref{mwpm:5x5} \\
    \textbf{$[[85,1,7]]$ }  & \ref{NOmst:7x7} & \ref{mst:7x7} & \ref{mwpm:7x7} \\
    
  \end{tabular}
}};

\end{tikzpicture}%
	} 
	\caption{ Logical error probability vs. physical error probability of the channel. Rotated $[[9,1,3]]$, and standard $[[13,1,3]]$, $[[41,1,5]]$, and $[[85,1,7]]$ surface codes over depolarizing channel.
		\label{Fig:plot_MST}}
\end{figure}
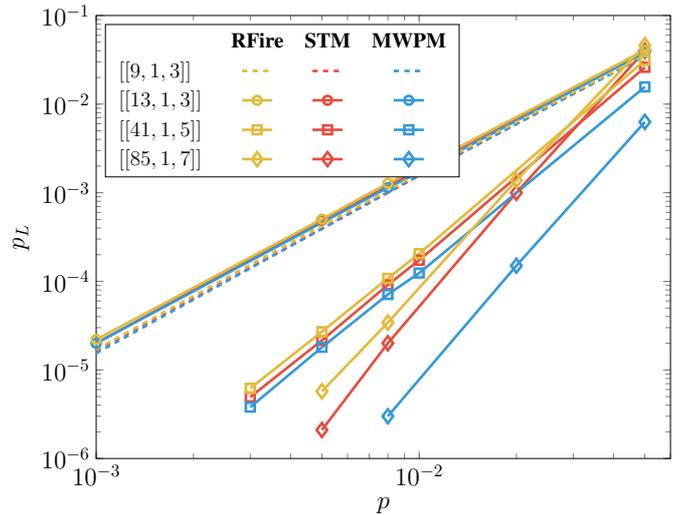

\begin{table*}[t]
    \centering
    \setlength{\tabcolsep}{3pt}
    \caption{AVERAGE EXECUTION TIMES $[\mu s]$}
    \label{tab:times}
    \small
    \begin{tabular}{lC{1.5cm}C{1.5cm}C{1.5cm}C{1.5cm}C{1.5cm}C{1.5cm}C{1.5cm}C{1.5cm}C{1.5cm}C{1.5cm}C{1.5cm}}
        \toprule 
        \rowcolor[gray]{.95}
        & & $n_\mathrm{d}=1$ & $n_\mathrm{d}=2$ & $n_\mathrm{d}=3$ & $n_\mathrm{d}=4$ & $n_\mathrm{d}=5$ & $n_\mathrm{d}=6$ & $n_\mathrm{d}=7$ & $n_\mathrm{d}=8$ \\
        \midrule
        \textbf{RFire} &  & $0.0008$ &$0.0015$ &$0.0021$ & $0.0025$ & & \\
        \textbf{STM} & $[[13,1,3]]$ & $0.0025$ &$0.0031$ &$0.0051$ & $0.0123$ & & \\
        \textbf{MWPM} &  & $17.329$ &$30.618$ &$62.254$ & $63.653$ & & \\
        \midrule
        \textbf{RFire} &  & $0.0019$ & $0.0044$ & $0.0051$& $0.0081$ & $0.0085$ & $0.0108$ & $0.0110$& $0.0341$  \\
        \textbf{STM} & $[[85,1,7]]$ & $0.0031$ & $0.0047$ & $0.0091$& $0.0121$ & $0.0695$ & $0.2071$ & $1.2839$& $1.9302$  \\
        \textbf{MWPM} &  & $19.156$ & $32.247$ & $55.864$& $67.591$ & $100.55$ & $115.04$ & $157.67$& $171.78$  \\
        \bottomrule
    \end{tabular}
\end{table*}

In this section we compare the performance of surface and rotated surface codes using \ac{MWPM}, \ac{STM}, and \ac{RFire} decoding via Monte Carlo simulations. 
These decoders are implemented in C++ and run with an Apple Silicon M2 processor.
In the implementation, we exploit the \ac{LEMON} C++ library for an efficient \ac{MWPM} algorithm \cite{DezBalJut:11}. 
To evaluate the complexity of the decoders we measure the average execution time of the processing that starts with the initial complete graph $\mathcal{G}$ (generation of $\mathcal{G}$ not included) and ends returning the solution $\mathcal{E}$.
The evaluation, carried out when varying the lattice size and the number of defects, is reported in Table~\ref{tab:times}.
As expected, the computational time saving is remarkable when \ac{STM} and \ac{RFire} are adopted.
Furthermore, in Fig.~\ref{Fig:plot_MST}, the logical error rate of some surface codes over a depolarizing channel is depicted.
In particular, the performance gap between \ac{MWPM} and the proposed fast decoders widens progressively as the distance of the code increases. This is due to the fact that, although the \ac{STM} and the \ac{RFire} decoders effectively correct all the errors with weights up to $t$, the \ac{MWPM} decoder corrects a larger fraction of those with weight $t+1$. 
In terms of speed, the \ac{RFire} decoder is the fastest, but it looses in error correction capability. 
Note that although we showcased our decoder using standard surface codes for the sake of clarity, its performance remains consistent when applied to rotated surface types, as illustrated in Fig.~\ref{Fig:plot_MST} for the $[[9,1,3]]$ code. 
Among the decoders available in the literature, we concentrate on comparing with the only current fast decoder, which is the \ac{UF} \cite{Iol:24}. To accomplish this, we utilized the Qsurface library \cite{Iol:24}. Specifically, for the $[[85,1,7]]$ surface code and a physical error rate $p_Z = 0.03$, the \ac{UF} shows a speedup by a factor of approximately $\cross 2$ with respect to \ac{MWPM}. 
In the same settings, the speedup of \ac{RFire} and \ac{MST} decoders with respect to \ac{MWPM} is $\cross 3625$, and $\cross 164$, resulting in the quickest execution times.

\section{Conclusions}\label{sec:conclusions}

We have introduced two fast decoders designed for surface codes and compared their performance with the widely used \ac{MWPM} decoder. While observing a slight reduction in error correction capability for errors of weight $j \geq t + 1$, our findings highlight a substantial advantage in terms of execution times, with a speedup of a factor $\cross 10.000$, for both the \ac{STM} and the \ac{RFire} decoders.



\bibliographystyle{IEEEtran}
\bibliography{Files/IEEEabrv,Files/StringDefinitions,Files/StringDefinitions2,Files/refs}


\end{document}